\documentclass{vldb}

\usepackage{listings}
\usepackage{ccaption}
\usepackage{algorithm}
\usepackage{algorithmic}
\usepackage{url}
\captionnamefont{\bfseries}
\captiontitlefont{\bfseries}

\usepackage{xspace,xcolor}

\begin{document}

\newcommand{\eat}[1]{}
\newcommand{\eg}{{\em e.g.}, }
\newcommand{\ie}{{\em i.e.}, }
\newcommand{\etal}{{\em et al.\ }}

\newcommand{\todo}[1]{\textcolor{red}{{\bf TODO:} #1}}

\title{Scalable Social Coordination using Enmeshed Queries}

\def\sharedaffiliation{%
\end{tabular}
\begin{tabular}{c}}

\numberofauthors{3}
\author{
\alignauthor
Jianjun Chen\\
\affaddr{Google Inc.} \\
\affaddr{Mountain View, CA, USA} \\
\email{jianjunc@google.com}
\alignauthor
Ashwin Machanavajjhala\\
\affaddr{Duke University} \\
\affaddr{Durham, NC, USA} \\
\email{ashwin@cs.duke.edu}
\alignauthor
George Varghese\\
\affaddr{Microsoft Research} \\
\affaddr{Mountain View, CA, USA} \\
\email{varghese@google.com}
}

\newcommand{\squishlist}{
   \begin{list}{$\bullet$}
    {
      \setlength{\itemsep}{0pt}
      \setlength{\parsep}{3pt}
      \setlength{\topsep}{3pt}
      \setlength{\partopsep}{0pt}
      \setlength{\leftmargin}{1.5em}
      \setlength{\labelwidth}{1em}
      \setlength{\labelsep}{0.5em} } }

\newcommand{\squishlisttwo}{
   \begin{list}{$\bullet$}
    {
      \setlength{\itemsep}{0pt}
      \setlength{\parsep}{0pt}
      \setlength{\topsep}{0pt}
      \setlength{\partopsep}{0pt}	
      \setlength{\leftmargin}{2em}
      \setlength{\labelwidth}{1.5em}
      \setlength{\labelsep}{0.5em} } }

\newcommand{\squishend}{
    \end{list}  }

\newtheorem{definition}{Definition}
\newtheorem{theorem}{Theorem}
\newtheorem{example}{Example}
\newtheorem{lemma}{Lemma}
\newtheorem{claim}{Claim}

\maketitle

\begin{abstract}
Social coordination allows users to move beyond awareness of their friends to efficiently coordinating physical activities with others.  While specific forms of social coordination can be seen in tools such as Evite, Meetup and Groupon, we introduce a more general model using what we call {\em enmeshed queries}.  An enmeshed query allows users to declaratively specify an intent to coordinate by specifying social attributes such as the desired group
size and who/what/when, and the database returns matching queries.  Enmeshed queries are continuous, but new queries (and not data) answer older queries; the variable group size also makes enmeshed queries different from entangled queries, publish-subscribe systems, and dating services.

We show that even offline group coordination using enmeshed queries is NP-hard.   We then introduce efficient heuristics that use selective indices such as location and time to reduce the space of possible matches; we also add refinements such as delayed evaluation and using the relative matchability of users to determine search order.   We describe a centralized implementation and evaluate its performance against an optimal algorithm.   We show that the combination of not stopping prematurely (after finding a match) and delayed evaluation results in an algorithm that finds 86\% of the matches found by an optimal algorithm, and takes an average of 40 $\mu$sec per query using 1 core of a 2.5 Ghz server machine.   Further, the algorithm has good latency, is reasonably fair to large group size requests, and can be scaled to global workloads using multiple cores and multiple servers. We conclude by describing  potential generalizations that add prices, recommendations, and data mining to basic enmeshed queries.
\end{abstract}

\section{Introduction}\label{sec:intro}
While social networks like Facebook are widely used for interacting with friends, they seem less useful for planning and {\em coordinating} with other people.  But, this is a significant part of what being {\em social} entails --- we go to birthday parties, we play sports, we fraternize with like-minded people on even obscure topics such as fly-fishing.  Further, the default method of coordinating with people by phone/email is tedious, and sometimes hard to drive to closure.  Thus social coordination seems the next step beyond social networking.

There are a number of ad hoc social coordination
tools already in the market.  These include \url{evite.com}  (coordinator specifies list of invitees and an event),  \url{doodle.com} (users specify a range of times for a meeting and the coordinator can look for common times), \url{meetup.com}  (users specify a zipcode and a topic, and then browse through the groups returned by the system), \url{foursquare.com}  (users check into a database of locations, and can query the
system for the location of their friends), and \url{groupon.com}  (users specify a location and interest, and then browse through a set of deals; however, deals only take effect when a critical mass of users sign up for the deal).

Existing social coordination tools have four disadvantages:
{\em 1. Ad Hoc:} Each system is tailor-made
to a specific application although they have many abstract features in common.  {\em 2. Limited Query Capability:} While most provide a rudimentary query facility, they also involve considerable browsing and selection by the user. {\em 3. Not Continuous:} If there is no current choice that satisfies the user, the user must retry later: synchronous polling is needed as opposed to asynchronous notification by the  system.
{\em 4. No group size control:}  Users cannot specify limits on the group size.

In contrast, our goal is to develop tools for {\em fluid} social coordination which have the following properties.
\squishlist
\item {\em Generic Platform:} We abstract coordination as finding matches on key attributes such as people, activity, location, and time in a common platform that many social applications can use.
\item  {\em Declarative Queries:} Users specify predicates that prescribe the kind of groups they wish: the system matches users to groups without need for browsing.
\item {\em Continuous:} If there is no current choice, the system will retain the user query and subsequently attempt to match this query when future queries enter.
\item {\em Group Constraints:}  Users can specify group constraints (such as bounds on group size).
\squishend
By fluid, we mean that coordination is not limited to static groups such as friends; instead, we allow coordinating even with strangers, and with different sets of people for different activities.  Granovetter~\cite{granovetter1973} has argued that novel information flow (such as job opportunities) typically occurs through ``weak" ties (e.g., casual acquaintances) than through ``strong" ties (e.g., friends and family).  We suggest that allowing people to perform fluid social coordination may encourage the formation of weak ties.

Fluid coordination appears in a number of real world applications. Consider multiplayer online gaming where users wish to form groups whose sizes can depend on the game.  The user is indifferent to player identities except that they be close by (to reduce latency), and have similar Internet access speeds and game ratings.   By allowing a user, of say Xbox LIVE, to specify parameters in  these three key attributes and a group size (say 4 to play Halo), our system can match users up.  As a second example, consider finding 3 partners among all Microsoft employees for playing doubles tennis in Sunnyvale at 4 pm on the weekend.
Private data collection is a third example. Users may allow a hospital to publish their information only if they share the same ``quasi-identifying'' attribute with a group of $k-1$ other users ($k$-anonymity \cite{sweeney2002}), or if additionally $\ell$ distinct diseases appear in their group ($\ell$-diversity \cite{Machanavajjhala2006}).

Complete fluidity may alarm some users, especially when considering physical activities such as tennis, where a user may not wish to play with
another player, say, with a criminal record.  We allow users effectively to ``scope'' their coordination intents using constraints. For instance, a user can request to only play with other users affiliated with Duke University (i.e., with a \url{duke.edu} email suffix). Users can set the scope
to be wide enough to match their query and yet narrow enough to stay in their comfort zone.

In this paper, we present a system that can support fluid coordination as described in the above examples. In this system, users describe their coordination intent using a novel concept called {\em enmeshed queries}, which allow a user to specify coordination constraints, like who/what/when, and constraints on coordinating group, like on the number of people forming a group. The formalization and scalable implementation of these queries is the main contribution of this paper.

From an algorithmic viewpoint, think of the set of enmeshed queries as forming the nodes of a graph. We place an edge from node $R$ to node $S$ if $R$ is {\em compatible} with $S$.  For example, the queries of two tennis players who have compatible ratings who wish to play at the same time and the same place will be linked by an edge.  The query system attempts to carve this graph into cliques such that each clique satisfies its group size constraints; for example, one clique could be a set of 4 compatible tennis players to play doubles.

\noindent{\bf Contributions and Paper Outline:} First, we formalize the problem of fluid social coordination by defining the concept of enmeshed queries.   Second, we show that optimal algorithms for matching enmeshed queries are NP-hard;  however, we introduce heuristics that  are close to optimal  on typical workloads.  Third, we describe an implementation that can scale to millions of concurrent enmeshed queries.

The rest of the paper is organized as follows. Section~\ref{sec:related} describes related work.  Section~\ref{sec:model} describes enmeshed queries, and theoretical hardness results are described in Section~\ref{sec:complexity}. Algorithms are described in Section~\ref{sec:algo}.  Section~\ref{sec:expt} presents an evaluation that suggests that our algorithms can scale to millions of queries per second. We conclude (Section~\ref{sec:conc}) with a list of future research directions. 

\section{Related Work}
\label{sec:related}

Enmeshed queries are most closely related to entangled queries~\cite{entangledS2011, KotGRGK10, entangledV2012}.  The seminal paper on entangled queries~\cite{entangledS2011} required users to explicitly specify partners they wish to coordinate with (e.g., tennis with John and Sarah).  More recent work~\cite{entangledV2012}  allows some flexibility in choosing partners by allowing an entangled query to specify {\em one} partner to be any friend (or someone related to the user through a prespecified binary relation). While enmeshed and entangled queries share the goal of declarative social coordination, there are significant differences. First, enmeshed queries allow group constraints such as bounds on group sizes; entangled queries do not. Second, enmeshed queries are designed for fluid coordination -- partners need not be explicitly specified, and indeed may not even be friends.   
Third, enmeshed queries adopt online algorithms that match a newly arrived query with a subset of existing queries. Our online algorithms are designed to scale to millions of concurrent queries. In contrast, entangled queries consider (offline) coordination among a given small set of queries (1000s of queries in \cite{entangledS2011, entangledV2012}).
Entangled queries, on the other hand, have considerable power in specifying constraints on external relations that may be the final object of coordination (e.g., a tennis reservation at a specified court).   An interesting future direction is to combine the power of these  two formalisms.

\begin{table*}[t]
\centering
{\small
\begin{tabular}{|c|c|c|c|c|}
\hline
& Group Constraints & Independent Queries  & Query Life &  \# of queries \\
\hline
Enmeshed & Yes & No  & Short & Millions \\
Entangled & No & No & Short & Thousands \\
Continuous Queries & No & Yes  & Long & Varying \\
Pub-Sub & No & Yes & Long & Millions \\
\hline
\end{tabular}
}
\caption{\label{fig-related}Comparison among continuous queries, pub/sub systems, entangled queries and enmeshed queries.}
\end{table*}

Pub-sub systems~\cite{pubsub00}  and continuous query systems~\cite{terry92} also provide declarative continuous query evaluation, but each query is logically independent; further there are no constraints on groups. On the other hand, enmeshed  queries find matches among queries.  Our approach to find candidate query matches is related to work in efficiently evaluating boolean expressions~\cite{fontoura10}. Finally, enmeshed queries differ from nested transactions \cite{lynch1988}, since at query time the system does not know which other enmeshed queries it is waiting on. Table~\ref{fig-related} summarizes these comparisons.

There is also related work on team formation in social networks~\cite{lappas09} that studies the following problem: given a set of people (i.e. nodes in a graph) with certain skills and communication costs across people (i.e. edges in the graph), and a task $T$ that requires some set of skills, find a subset of people to perform $T$ with minimal communication costs. They prove that such problems are NP hard, and describe heuristics to reduce computation. Our problem is not the same as task formation because tasks are not explicit first-class entities in enmeshed queries but are, instead, implicit in the desires of users.  Further, our metric is maximizing matches and not minimizing communication.

\section{Model and Preliminaries}
\label{sec:model}
\newcommand{\users}{\ensuremath{{\cal U}}}
\newcommand{\queries}{\ensuremath{{\cal Q}}}
\newcommand{\A}{\ensuremath{{\cal A}}}
\renewcommand{\S}{\ensuremath{{\cal S}}}
\newcommand{\J}{\ensuremath{{\cal J}}}
\newcommand{\G}{\ensuremath{{\cal G}}}
\newcommand{\card}{\ensuremath{\mbox{card}}}
\newcommand{\comp}{\ensuremath{\;\theta\;}}
\newcommand{\that}{\ensuremath{\mathbf{that}}}
\newcommand{\N}{\mathbb{N}}
\newcommand{\alg}{\ensuremath{\mbox{\sc alg}}}
\newcommand{\opt}{\ensuremath{\mbox{\sc opt}}}

We formally define enmeshed queries and the problem of social coordination.

\noindent{\bf Users}: Consider a set of users $\users$ who wish to participate in coordination tasks. Each user is associated with a set of attributes $\A_\users$.  Examples of user attributes include age, home location, tennis rating, etc. Denote by $dom(A)$ the domain of an attribute $A$, and by $dom(\A) = \times_{A \in \A} dom(A)$ the cross product of the domains of a set of attribtues $\A$. A point $\vec{x} = [x_1, \ldots, x_k]$ is a multidimensional value from $dom(\{A_1, \ldots, A_k\})$.

\noindent{\bf Enmeshed Queries}: Users specify their intent to coordinate using one or more enmeshed queries.
An enmeshed query $q$ is associated with (i) a unique user $q.user$, and (ii) a set of free coordination variables
$q.\vec{x} = [q.x_1, q.x_2, \ldots, q.x_{|C|}]$. Each free variable $q.x_i$ takes values from a distinct coordination attribute $A_i \in \A_C$.  Examples of coordination attributes are time(when), location (where), activity type (what), etc. Enmeshed queries define coordination intent by specifying constraints. Formally, an enmeshed query is a triple $(\S, \J, \G)$, where $\S$ is a set of selection constraints, $\J$ is a set of join constraints, and $\G$ is a set of group constraints.  We define $\S$, $\J$, and $\G$ in turn.

First, $\S$ specifies the query's coordination intent $\S$ by specifying sets of possible values for each coordination variable. For instance, a user may want to play tennis or squash, between 4 and 6 PM, in either Cupertino or Mountain View. More formally, coordination intent
$\S$ is specified as a conjunction of {\em selection constraints}, where each selection constraint is of the form $x_i \in S$, where  $S \subseteq dom(A_i)$. Note that queries may not pose any constraint on some coordination attributes; this is modeled using $S = dom(A_i)$.

Second, $\J$ specifies additional {\em join constraints} over the user attribute values on pairs of enmeshed queries. Examples include:  Alice would like to coordinate only with her friends, and Alice would like to play tennis with users who have a higher rating.  If Alice knew a priori that she was coordinating with Bob, then Alice could easily express the above rating constraint as $Bob \in Alice.Friends$ and  $(Bob.rating > Alice.rating)$. However, when declaring coordination intent, Alice does not know who she is coordinating with. To describe join constraints, we introduce the notation $\that$ to refer to attributes of other users who may potentially coordinate on a task. Thus the above constraints in Alice's query $q$ can be written as  $\that.id \in user.Friends$ and $\that.rating > user.rating$.  $\J$ is a conjunction of such individual join constraints.

Finally, $\G$ specifies predicates on allowable groups.  For example, one could specify the average rating of the group of individuals for a doubles tennis game.  For this paper, we focus on the simplest and most useful group constraint, a {\em cardinality constraint} such as ``want to play tennis with at least 2 or at most 4 individuals''. Formally, a cardinality constraint $\G$ is expressed as $\card(q) \in [lb, ub]$, where $[lb, ub]$ ($lb \geq 2$) is the set of integers $\leq ub$ and $\geq lb$.

\noindent{\bf Example}: 
A user wanting to play tennis in Cupertino or Sunnyvale at 8 PM with 2,3 or 4 people with a rating 5 can be written as the following query:
\begin{eqnarray*}
& A_{time} \in [8PM] \wedge A_{location} \in \{Cupertino, Sunnyvale\} \wedge & \\
& A_{sport} \in \{Tennis\} \wedge \that.A_{rating} = 5 \wedge card(q) \in [2,4]&
\end{eqnarray*}

\noindent{\bf Matching Queries}:  We define the semantics for how/when coordination is achieved in stages by adding in each type of constraint in turn.
First, we say that a set of queries $q_1, \ldots, q_k$ are {\em jointly satisfied} if there is a point $\vec{p} \in dom(\A_C)$ such that for every query $q_i$ and selection constraint $x_j \in S_{ij}$ in $q_i$, we have $p_j \in S_{ij}$. Next, we incorporate join constraints as follows.  Two queries $q_1$ and $q_2$ are said to {\em match} if (i) their selection constraints are jointly satisfied, and (ii) join constraints on $q_1$ and $q_2$ are satisfied.  Finally, we add group cardinality constraints:

\begin{definition}[Committable Query Group]
A group of queries $Q \subseteq \queries$ is called a {\em committable query group} if
\squishlist
\item Selection conditions on queries in $Q$ are jointly satisfied,
\item $\forall q_1, q_2 \in Q$, $q_1$ and $q_2$ match on join constraints, and,
\item $|Q|$ satisfies the cardinality constraint of all $q \in Q$.
\squishend
\end{definition}
A committable group can be returned by the system as a valid set of users who can be matched together. All queries in such a group are called {\em committed}.

\noindent{\bf Problem Statement}:  User coordination is now reduced to the problem of finding committable query groups. Note that the problem is online -- the system cannot wait till all the queries have been submitted.

\begin{definition}[Group Coordination Problem]\ \\
Given a stream of enmeshed queries $\queries = q_1, q_2, \ldots, q_n$, find sets of committable query groups such that the number of committed queries is maximized. We denote by $\opt(\queries)$ the {\em offline optimal}, or the maximum number of committable queries that can be committed if all the queries are known upfront. Our goal is to design an algorithm $\alg$ such that for any finite
subsequence $Q$ of $\queries$, the number of queries committed ($\alg(Q)$) is as close to $\opt(Q)$ as possible.
\end{definition}

While maximizing the throughput (number of committed queries) is our main goal, we will also experimentally measure query processing time, latency (measured from when a query enters to when it was actually committed as a group) and fairness with respect to group size of our algorithms.

\section{Complexity}
\label{sec:complexity}
We show complexity results for the Offline and Online group coordination problems.
\subsection{Offline Problem}
In this section we show that finding the offline optimal $OPT(\queries)$ for the
group coordination problem is NP-hard, whether we want to maximize the number
of committed queries or groups. When maximizing queries, or groups, we show
the problem is NP-hard even when there is a single coordination attribute.

\begin{lemma}
Given a set of enmeshed queries $\queries$ with a single coordination attribute $A$,
computing the maximum number of committed groups is NP-hard.
\end{lemma}
\begin{proof}(sketch)
We show hardness via a reduction to the well known NP-hard problem of
finding the largest independent set in a graph \cite{gareyj:np}.

Given  an instance $G = (V,E)$ of the max independent set problem,
we construct an instance of the group commit problem as follows. There is
a single attribute $A$, which has one value $a_v$ for every node $v \in V$; i.e.,
$|A| = |V|$. The set of queries $Q$ contains one query $q_e$ for every edge
$e = (u,v) \in E$, with constraints $(A \in \{u, v\} \wedge cardinality \in
\{d_u, d_v\})$, where $d_u$ is the degree of the node $u$.

Note that in this construction, the only committable groups are
sets of edges $S_u$ that are incident on a node $u$ in the graph $G$. Moreover,
if $S_u$ commits, then for all $v$ adjacent to $u$, $S_v$ can not commit.
Therefore, every feasible set of committed groups corresponds to an independent
set in the graph.
\end{proof}

\begin{lemma}
Given a set of enmeshed queries $\queries$ with a single coordination attribute $A$,
computing the maximum number of committed queries is NP-hard.
\end{lemma}
\begin{proof}(sketch)
We show hardness via a reduction to the maximum 3-dimensional matching problem,
a classic NP-complete problem \cite{gareyj:np}. An instance of the 3-D matching
problem consists of disjoint sets $X, Y, Z$ and a subset $T \subseteq X
\times Y\times Z$ of triples. $M \subseteq T$ is called a {\em matching} if
for any two distinct triples $(x,y,z), (x',y',z') \in M$, we have $x \neq x',
y \neq y'$ and $z \neq z'$.

Given an instance of 3-D matching, we construct an instance of the group commit
problem as follows. There is a single coordination attribute $A$ whose domain is
$X \times Y \times Z$. Without loss of generality, we can assume that every element
in $v \in X \cup Y \cup Z$ appears in at least one triple in $T$ (otherwise 0-degree
elements can be removed from the problem). With every element $v$, we associate a
query $q_v$ having a cardinality constraint of $3$, and a coordination $A \in S$,
where $S$ is the set of triples that contain $v$.

Every committable group corresponds to a unique triple $(x,y,z)$. Moreover, since no
two committed groups can share a query, the set of committed groups gives a 3-D matching.
Finally, the proof follows from the fact that the number of queries is just 3 times the
number of committed groups.
\end{proof}

\subsection{Best-Effort Online Coordination}
In this section, we consider a special class of online algorithms. An algorithm is {\em best-effort} if when a query $q$ enters the system at most one committable group is returned and any returned committable group includes $q$.  It may not return a committable group even when one exists.  By contrast, an algorithm is {\em optimal best-effort} if when a query $q$ enters the system exactly one committable group is returned that includes $q$, if such a committable group exists.

Best-Effort algorithms are attractive because they ensure constant progress. Second, best-effort algorithms are $k$-competitive, where $k$ is the maximum cardinality constraint of a query.  That is, the number of queries committed by a best-effort algorithm is at least $1/k$
times the number of queries committed by an offline optimal algorithm. The competitive result follows because for any group output by the best-effort algorithm, in the worst case, the offline optimal might have used each query in the group to commit a separate group.   However, {\em optimal} best-effort algorithms are hard to design.

\begin{figure}[t]
\centering
\includegraphics[width=2in]{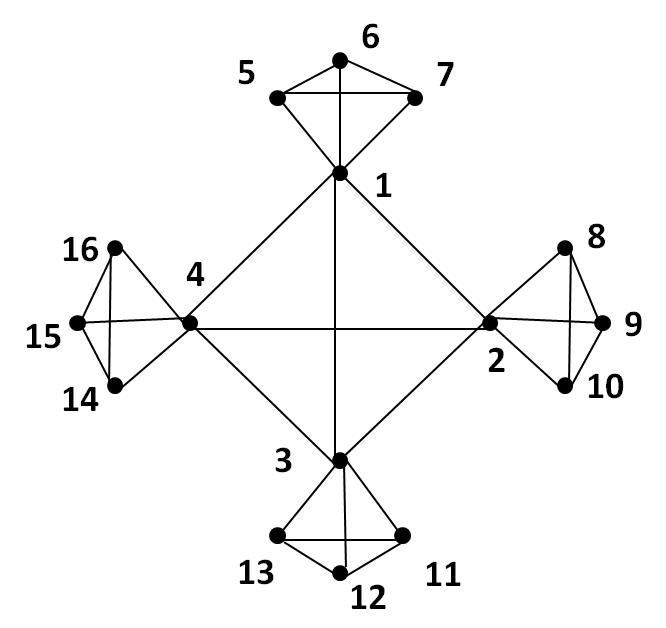}
\caption{\label{fig:competitive} Example illustrating the competitive ratio of best-effort algorithms. Nodes are queries, edges represent whether or not queries match. All queries have cardinality constraint $=4$. The numbers also represent the order in which queries enter the system. }
\end{figure}

\begin{lemma}
Best-effort algorithms are $k$-competitive.
\end{lemma}
\begin{proof}
The following construction shows that the competitive ratio is at least $k$. Consider $k^2$ queries which are numbered $q_1, \ldots, q_{k^2}$ in the order they enter the system. Queries $q_1, \ldots, q_k$ form a committable group. Additionally, $q_i$ and $q_{i(k-1) + 2}, \ldots, q_{(i+1)(k-1)+1}$ form $k$ different committable groups for $i = 1$ to $k$.
Figure~\ref{fig:competitive} shows the construction for $k = 4$; nodes are queries, edges represent whether or
not queries match. All queries have cardinality constraint $=4$. More specifically,
for queries 1-4, $q_i$'s constraint is $A \in \{x, i\}$. For queries connected
to $q_i$ ($i = 1, 2, 3, 4$), the constraint is $A \in \{i\}$.

An optimal solution is to commit $(1,5,6,7)$, $(2,8,9,10)$, $(3,11,12,13)$ and
$(4,14,15,16)$ (the $k$ groups in the general case). However, a best effort algorithm would greedily commit $(1,2,3,4)$, ($q_1, \ldots, q_k$ in the general case). Thereafter no other query can commit. Thus the competitive ratio is at least $k$.

For any group output by the best-effort algorithm, in the worst case, the offline
optimal might have used each query in the group to commit a separate group. Since the
max cardinality constraint is $k$, the competitive ratio is $\leq k$.
\end{proof}

However, we next show that best-effort algorithm also are hard to design. When a new query comes in, best-effort algorithms need to find a committable group if there is one. But, this problem is NP-hard in the presence of join constraints. In the absence of join constraints, we present a PTIME best-effort algorithm which will be the scaffolding for our efficient algorithms presented in the next section.

\begin{lemma}
Computing an optimal best-effort solution is NP-hard, even if the number of attributes is a constant. Suppose the number of attributes is a constant. Ensuring that an algorithm always finds a committable group if there is a committable group is NP-hard.
\end{lemma}
\begin{proof}(sketch)
Let ${\cal S}$ be the set of query groups that can be committed due to adding $q$ to $\queries$. Suppose there do not already exist any committable groups in $\queries$. Hence, for all groups $S \in {\cal S}$, $q \in S$.

It is easy to verify that the problem is in NP. Given a group, we can efficiently check whether it is committable.

To show it is hard we present a reduction from the clique decision problem: the problem of determining whether there is a clique of size $k$ in a graph $G = (V,E)$ is well known to be NP-hard.  Given an instance of the clique problem, we construct an instance of our problem with a single attribute $A$. For every edge $e = (u,v) \in E$, we have a value $a_{e} \in A$. We use $e$ to represent both $(u,v)$ or $(v,u)$.

We construct one query $q_v$ for every $v \in V$, and each query has cardinality equal to $(k+1)$. {\em The queries are such that $(q_v, q_u)$
matches if and only if $(v,u) \in E$.} This can be achieved just by constructing  an ID attribute $A_{ID}$ and a ``friend'' attribute $A_{friend}$. For a query $q_v$, $A_{ID} = v$ and $A_{friend} = \{u | (u,v) \in E\}$. Finally, we include  a join constraint, $\that.A_{ID} \in A_{friend}$ . This require that each query $q_v$ be satisfied only by a $q_u$m, where $(u,v) \in E$.

We have a new query $q$ that matches all queries we constructed (by not having any constraint for $q$) and also having a cardinality $k+1$. If we find a committable group in the new instance, we get a clique in the original graph. Thus the problem is NP-hard.
\end{proof}

In the absence of join constraints, there is a simple best-effort algorithm (Algorithm~\ref{algo:complete}) to determine if there is a committable group in the graph when a new query $q$ enters the system. We first identify all points $p$ in $dom(\A_C)$ that satisfy $q$. For each $p$, we compute the set $S_p$ of queries which are jointly satisfied by $p$. Finally, we can efficiently (in time linear in $|S_p|$) compute a committable group of queries (which satisfy the cardinality constraints) as follows: divide the number line into segments determined by the upper and lower bounds on cardinality for all the queries in $S_p$. For each segment $(l, u)$, compute the queries $S_{(l,u)}$ that permits groups of size between $l$ and $u$. If $l \leq |S_{(l,u)}| \leq u$, then $S_{(l,u)}$ can be returned as a committable group. If there is no such group, then the result is empty. Algorithm~\ref{algo:complete} takes $O(|Q|\cdot|P|)$ time, where $P \subseteq dom(\A_C)$ is the set of points that satisfy $q$.

\begin{algorithm}[t]
\caption{Optimal Algorithm (No Join Constraints)}
\begin{algorithmic}[1]
\STATE {\bf Input:} A set of enmeshed queries $Q$, a new query $q$
\STATE {\bf Output:} $true$, if $\exists$ committable group containing $q$,
\STATE \hspace{1.5cm}$false$, otherwise.
\STATE Let $P \subseteq dom(\A_\users \cup \A_C)$, s.t., $\forall p \in P$, $p$ satisfies $q$
\FOR {$p \in P$}
\STATE {\em // STEP 1: Identify queries whose coordination constraints are satisfied by $p$.}
\STATE Let $S_p$ be the set of queries satisfied by $p$.
\STATE
\STATE {\em // STEP 2: Identify a subset that satisfies the cardinality constraints.}
\STATE Let $C$ denote a sorted list of numbers between $lb_q$ and $ub_q$ (inclusive) that also correspond to cardinality limits of queries in $S_p$
\FOR {consecutive numbers $c_i, c_{i+1} \in C$ }
\STATE Let $S_i$ denote the set of queries which allow coordination with size in $[c_i, c_{i+1}]$.
\IF {$c_i \leq |S_i| \leq c_{i+1}$}
\RETURN $true$
\ENDIF
\ENDFOR
\ENDFOR
\RETURN $false$
\end{algorithmic}
\label{algo:complete}
\end{algorithm}

\section{Scalable Group Coordination}
\label{sec:algo}
Given that optimal best-effort algorithms are NP-hard, we  present four heuristic algorithms that will work well on workloads satisfying the following assumptions:

{\em A1. High Selectivity:} We assume most queries are selective, and hence, can be indexed based on some coordination attributes. First, people know what they want: queries like ``play any sports at any time and anywhere'' are not typical. Second, selectivity can be enforced by application user interfaces that disallow wildcards, and enforce limited disjunction (e.g., at most two locations allowed in a query).

{\em A2. Small Group Coordination:} We also assume that the average cardinality of queries are small, e.g. $<50$. 
This ensures that queries do not wait too long to be committed.

All our algorithms use the following structure.   They all start by using a subset of the most selective attributes in a query as an index to quickly find a (hopefully small) ``short list'' of potentially matching queries.  For example, out of a million concurrent enmeshed queries, only 50 queries in the system may specify `playing Tennis in Cupertino at 3 PM next Tuesday'.   While the first pruning step only considers selection constraints, the second step traverses the ``short list'' to greedily attempt to build a committable group while also incorporating join and cardinality constraints.

\subsection{ Basic Matching Algorithm (BMA)}
\label{sec:BMA}

Basic matching algorithm, BMA (Algorithm~\ref{algo:match}), is a best-effort algorithm that traverses the ``short list" in random order, and stops when it finds the first committable group.

\newcommand{\qt}{\mbox{QueryTable}}
\newcommand{\pt}{\mbox{\sc Pindex}}
\newcommand{\gt}{\mbox{\sc Cindex}}

\begin{algorithm}[t]
\caption{BMA: Basic Matching Algorithm}
{\small
\begin{algorithmic}[1]
\STATE {\bf Input:} A set of enmeshed queries $Q$, a new query $q$, an inverted index $\pt$
\STATE {\bf Output:} A committable group $C$ (and $Q \leftarrow Q - C$),
\STATE \hspace{1.5cm}  or null (and $Q \leftarrow Q \cup \{q\}$).
\STATE Let $P \subseteq dom(\A_\users \cup \A_C)$, s.t., $\forall \vec{p} \in P$, $\vec{p}$ satisfies $q$
\FOR {$\vec{p} \in P$}
\STATE {\em // STEP 1: Identify queries whose selection constraints are satisfied by $\vec{p}$.}
\STATE $PQS_p \leftarrow \pt.lookup(\vec{p})$
\STATE $CQS_p \leftarrow$ subset of queries in $PQS_p$ satisfying $\vec{p}$
\STATE
\STATE {\em // STEP 2: Identify subset of queries whose join and cardinality constraints are satisfied}
\STATE Set $S_p \leftarrow \{q\}$
\FOR {$q' \in CQS_p$ {\em // in random order}}
\IF {($q'$ matches $\forall q \in S_p$)}
\STATE $S_p \leftarrow S_p  \cup \{q'\}$
\STATE $C \leftarrow \mbox{\sc FindCommittableGroup($S_p$)}$
\IF {$C$ is not null}
\STATE Remove all queries in $C$ from $Q$ and $\pt$.
\STATE Return $C$
\ENDIF
\ENDIF
\ENDFOR
\ENDFOR
\STATE {\em // No committable group found}
\STATE Add $q$ to $\qt$ and to $\pt$.
\end{algorithmic}
}
\label{algo:match}
\vspace{-1mm}
\end{algorithm}

Given a query $q$, BMA iterates over all possible points $\vec{p} \in \A_C$ that satisfy the selection constraints of $q$ in random order. For each $\vec{p}$, it first finds a set of queries ($PQS_p$) that partially match $q$ using $\pt$, which is an in-memory inverted hash index defined over a subset of coordination attributes (e.g., location and time). This index (built using techniques in \cite{fontoura10}) can support hierarchical values, e.g. weekday instead of dates.  Since $\pt$ is only on a subset of attributes, not all queries in $PQS_p$ are jointly satisfied by satisfy $\vec{p}$. Hence, next a subset of queries satisfying $\vec{p}$, called $CQS_p$, is computed. $CQS_p$ is the ``short list'' we referred to earlier.  If assumption {\em A1} holds, $|CQS_p|$ will be small.

However, not all pairs of queries in $CQS_p$ may match according to the join constraints. Hence, BMA next greedily computes $S_p \subseteq CQS_p$, a subset of queries that match (and satisfy $\vec{p}$) as follows: Starting with $S_p = \{q\}$, BMA iteratively adds a randomly chosen query $q' \in CQS_p$ to $S_p$ if it matches every other query in $S_p$ (based on join constraints). As the following example illustrates, there may be many choices for $S_p$ for the same $CQS_p$.

\begin{example}\label{ex:greedy}
Consider queries with a user attribute $A_{rating} \in \{1, 2, \ldots, 5\}$. Let $q$ be a query with no join constraints, with $A_{rating} = 3$. Suppose $CQS_p$ consists of $q_1, q_2, q_3$, such that $q_1.A_{rating} = 4$, $q_2.A_{rating} = q_3.A_{rating} = 2$, and $card(q_2) = card(q_3) = 3$. Suppose $q_1$ has a join constraint $\that.A_{rating} < A_{rating}$ (only play with lower ranked players), and $q_2, q_3$ have join constraints $\that.A_{rating} = A_{rating}+1$ (only play with players ranked 3). Then if the queries are considered in the order $q_1, q_2, q_3$, the resulting $S_p = \{q, q_1\}$. If $q_2$ or $q_3$ is considered before $q_1$, then the resulting $S_p = \{q, q_2, q_3\}$.
\end{example}

Every time a new query is added to $S_p$, the {\sc FindCommittableGroup} subroutine attempts to find a committable group that satisfies cardinality constraints as follows: Every query added to $S_p$ is also added to an inverted index from possible group sizes to queries, called $\gt$. A query with multiple group sizes appears multiple times in $\gt$, one for each group size. BMA iterates over $\gt$ to check whether there is some $m$ such that there are $\geq m$ queries in $S_p$ that permit a group of size $m$. BMA picks some group of $m$ queries. Assumption {\em A2} ensures that the average number of iterations over $\gt$ is small. If a group is formed, all queries from the group are removed from the system. Otherwise, we add current query q into the $\gt$.

\begin{example}\label{ex:alternatives}
For instance, suppose $q_1, q_2, q_3, q_4, q_5$ are matching queries in $S_p$ with cardinality  constraints $2, [2,3], [2,3], [2,3]$ and $3$ respectively. Then $\gt$ will contain $2$ entries: $2 \rightarrow \{q_1, q_2, q_3, q_4\}$ and $3 \rightarrow \{q_2, q_3, q_4, q_5\}$. Any pair from  $\{q_1, q_2, q_3, q_4\}$ or any subset of size 3 from $\{q_2, q_3, q_4, q_5\}$ are committable. {\sc FindCommittableGroup} will return some size 3 subset of $\{q_2, q_3, q_4, q_5\}$.
\end{example}

\subsection{Extensions to BMA}
BMA randomly traverses the ``short list" and returns the first committable group it finds. The next three algorithms extend BMA by more carefully choosing a committable group (among many alternatives), delaying committing to improve throughput, and by using more intelligent traversal orders.

\subsubsection{NES: No Early Stop}
BMA may choose to commit a group even before all the queries in $CQS_p$ are seen.  NES, on the other hand, computes $S_p$ by considering {\em every} query in $CQS$, and then calls {\sc FindCommittableGroup} to find the largest subset of $S_p$ satisfying cardinality constraints.  In Example~\ref{ex:alternatives}, suppose the queries in $S_p$ are considered in order of their indices. BMA will commit only 2 queries, $(q_1, q_2)$ after seeing the first two queries. On the other hand, NES waits to see all queries and commits $(q_2, q_3, q_4)$.

\subsubsection{DELAY: Delayed Matching}
Both BMA and NES are best-effort -- when a new query arrives, if a group can be committed, it will be committed. However, delaying matching smaller groups may help find a larger group later.

\begin{example}\label{ex:delay}
Consider again Example~\ref{ex:alternatives}. Suppose queries come in the order of their index. After $q_1$ and $q_2$ enter the system, any best-effort algorithm would commit the group $(q_1, q_2)$. Similarly, after $q_3$ and $q_4$ enter the system, any best-effort algorithm would commit the group $(q_3, q_4)$. However, if we waited till $q_5$ came into the system, NES  can commit $(q_3, q_4, q_5)$, thus committing all queries.
\end{example}

We implement DELAY as follows.  Let the {\em age} of a query denote the time elapsed since it was inserted into the system\footnote{We assume each query is inserted at a new time instant.}. When a query first arrives at time $t$, we compute its $CQS_p$, and proceed to find a committable group only if the average age of the $CQS_p$ is greater than a threshold $\tau$. Otherwise, $q$ is inserted into a scheduler, akin to a timing wheel, to be reevaluated at time $t+ \tau$. A query is reevaluated at most once. DELAY can increase processing time and query latency, as some queries are evaluated twice.  If at most $m$ queries simultaneously enter the system, DELAY returns no more than $2m$ committable groups: $\leq m$ groups from new queries and $\leq m$ from delayed queries.

\subsubsection{Query Matchibility}
\noindent All of the previous algorithms may choose to commit a group that is not optimal in terms of future queries.

\begin{example}
Consider queries that have an attribute $A_{where}$ with domain $\{Cupertino, Sunnyvale, Campbell\}$. Suppose $A_{where} = Campbell$ is very rare. Consider a new query $q$ with $A_{where} = Campbell$ and cardinality constraint $[2,3]$. Suppose its $CQS$ contains three queries $q_1, q_2, q_3$ with $q_1, q_2$ having $A_{where} \in \{Cupertino, Campbell\}$ and $q_3$ requiring $A_{where} \in \{Campbell\}$, resp. Moreover, $q_1, q_2$ have cardinality constraint $3$, while $q_3$ requires $2$. Then, returning the group $(q, q_3)$ is better than returning $(q, q_1, q_2)$ (even though the latter is larger), since it is more likely that a Cupertino query will arrive in the future than a Campbell query.
\end{example}

We call the propensity of a query to be matched as its {\em matchibility}, and is intuitively the expected number of other queries that will match $q$. Intuitively, a query has  low matchibility either because its selection or join constraints are hard to satisfy, or because it has large cardinality.  We use a simple online approach to estimate matchibility.

We define the matchibility of a query $q$ as the number of times $q$ appears in a $CQS_p$ before it is committed. Intuitively, a query that appears more frequently in the $CQS_p$ of other queries should have a higher chance to find matches. The initial value of a query's matchibility is
$$m(q) = n_q + c/lb_q$$
where $c$ is a constant, $lb_q$ is the lower bound on the cardinality constraint in the query, and $n_q = ub_q - lb_q + 1$ is the number of possible cardinality values that satisfy $q$. For example, if a query has a cardinality constraint $[4,6]$, then its initial matchibility is set to $c/4 + 3$. The motivation is that smaller groups are more matchable; further, a wide range of cardinalities is more matchable. Every time $q$ appears in some other query's $CQS_p$, $m(q)$ is incremented by 1.

We consider two simple heuristics -- {\em high matchibility first (HMF)} and {\em low matchibility first (LMF)} that traverse the ``short list" $CQS_p$ in matchibility order. Intuitively, LMF may match more queries than HMF by reserving popular queries for later consideration.

\noindent{\bf Discussion:}
The 3 extensions presented in this section are orthogonal, and an algorithm can be implemented with some or all of them together. However, given that DELAY postpones queries for later consideration, it seems unreasonable to not consider all queries. Hence in our implementation, DELAY always implies NES, but not vice versa.

\section{Evaluation}
\label{sec:expt}

In this section we evaluate our algorithms on synthetically generated workload for an example sports coordination application.
\subsection{Example Application}
We describe our evaluation in terms of an example sports coordination application. In this application, users want to find sports partners
that live nearby and with similar ratings. There are two user attributes -- home location, rating. There are four coordination attributes -- location, time, sport and opponent rating. Each user specifies selection constraints on the desired location(s), desired time(s), and sport. Additionally, a user may specify what the opponent's rating must be -- this can be captured using a join constraint. For instance, a user wanting to play tennis in Cupertino or Sunnyvale at 8 PM with 2,3 or 4 people with a rating 5 can be written as the following query:
\begin{eqnarray*}
& A_{time} \in [8PM] \wedge A_{location} \in \{Cupertino, Sunnyvale\} \wedge & \\
& A_{sport} \in \{Tennis\} \wedge \that.A_{rating} = 5 \wedge card(q) \in [2,4]&
\end{eqnarray*}
\subsection{Metrics and Setup}
We implemented all the algorithms presented in the paper, namely BMA, NES, DELAY, LMF and HMF. We compare our group coordination algorithms using system measures, throughput and  query processing time, as well as user measures, latency and fairness. These metrics are defined in Section~\ref{sec:model}. We use a single machine that is a 2 x Xeon L5420 2.50GHz running 64bit RHEL Server 5.6 with 16 GB memory.

\begin{table}[t]
\centering
{\small
\begin{tabular}{|c|c|c|c|}
\hline
Parameter & Domain & Multi- & Distribution  \\
Name &  & valued & \\
\hline
Home Location & $\{1, \ldots 10\}^2$ & 1 & Zipfian \\
Sport & $\{1, \ldots 10\}$ & $\leq 3$ & Uniform \\
Rating & \{1, \ldots 5\} & 1 & Zipfian \\
\hline
\end{tabular}
}
\caption{Parameter distribution for generating users in our synthetic workload.}
\label{tab-users}
\end{table}

\begin{table}[t]
\centering
{\small
\begin{tabular}{|c|c|c|c|}
\hline
Parameter & Domain & Multi- & Distribution  \\
Name &  & valued & \\
\hline
Location & $\{1, \ldots 10\}^2$ & $\leq 2$ & Zipfian \\
Time & 28 days $\times$ 12 hrs & $\leq 2$ & Bimodal \\
Sport & $\{1, \ldots 10\}$ & $1$ & Uniform \\
Rating & $\{1, \ldots 5\}$ & $\leq 2$ & Zipfian  \\
Group Size & $\{2, \ldots, 12\}$ & $1$ & Zipfian\\
\hline
\end{tabular}
}
\caption{Parameter distribution for generating queries in our synthetic workload.}
\label{tab-queries}
\end{table}

\subsection{Synthetic Workload}
\label{sec:workld}

Our default data trace contains 1 million
queries for 200,000 unique users.  Each user has have about 5 queries in
the workload since each user has a equal chance to be chosen for a query.  We also
enforce that no more than one query from the same user can share a time slot.
Each query in the workload contains a monotonically increasing
time stamp, desired location(s), desired time(s), action, desired ratings for
potential matching candidates and a range of desired group size.  We now specify
details for each parameter. The parameters are also summarized in Tables~\ref{tab-users} and \ref{tab-queries}.  All Zipf distributions
use an exponent parameter of $1$ where the second
most common frequency occurs 1/2 as much as the first, etc.

{\em Location:} Our workload generator generates 100 locations, represented by a two
dimensional array location[10][10].  Each user is assigned a location (i.e.
location[k][j]) as his/her home location based on a Zipfian distribution that models the
fact that some locations (e.g., San Francisco) have more users than others (say Brisbane,
CA). Besides the home location, a user query can also choose from among 4 neighboring locations: if a user is at location $[k],[j]$, the neighboring locations
are $[k+1][j]$, $[k-1][j]$, $[k][j+1]$, $[k][j-1]$.
This models the fact users that only want to play sports in locations "near" their
home locations. We allow at most 2 locations specified in a query using a Zipfian distribution wherein it is more probable to generate queries with 1 location (always home location) than 2 alternative locations
(home location + 1 choice made uniformly among the 4 neighbors).

{\em Time:} Times are represented as hourly slots, e.g. 2/1/2012 3pm. The domain is
between 8am-8pm for the following 4 weeks. We assign higher probability when
generating slots on weekends than those on weekdays. Times in a query are chosen uniformly
from slots in the domain based on their relative probabilities. No more than 2 time slots
can be specified in a query; as in the case of location, we use a Zipfian to choose the
number of time slots specified in a query, with a higher probability of a query
being generated with 1 time slot than with 2 time slots.

{\em Actions and Ratings:} Each user has a set plays up to three
sports with the exact number being chosen using a Zipfian distribution, with one being the most probable.  Once the number of
sports is chosen, the specific sports played by a user is chosen uniformly at random
from a set of 10 sports.  We also assign a rating ($r$) for each
chosen sport to represent a user's skill level, again using a Zipfian
distribution.  Similar to the case of location, we limit each query
to have no more than 2 desired ratings, where most queries will have 1 rating
(equal to $r$), some have two ratings ($r$ and either $r+1$ or $r-1$) where the frequency of 1 versus 2 is chosen by a Zipfian.
While a user can play at most 3 sports, each user query picks a {\em single} sport
uniformly at random from among the set of sports the user plays.

{\em Group sizes:} Group size generation is described in the next section.

\subsection{Finding a yardstick for optimality}
Given that finding an optimal solution to enmeshed queries is NP hard, a major challenge is measuring how
close our matching algorithms come to the optimal solution. In addition, since the total
space of attributes (e.g., the cross-product of location, time etc) is large,  if we
get a low percentage of queries matched by our algorithms, we cannot determine
whether the low match percentage
is caused by ineffective algorithms or because the workload generated
sparse points in a large domain space. To solve this dilemma, as we said earlier,
we generate data traces that
only contains queries belong to pre-matched groups. The advantage of this approach
is that we know an optimal solution is able to find matches for all queries in the
workload, which then serves as a yardstick for our heuristic algorithms.

More precisely, to generate queries with pre-matched groups, we first generate a group
signature that contains a single value for each dimension and a group size ($k$).
We then select a set of qualified users using an inverted index that maps from (location,
sport, rating) to users. Next, we generate $k$ queries that are compatible with the group
signature from the set of qualified users by following distributions described in
the workload generation description.

We generate a query's group size as follows.
A pre-match group size $m$ is first generated using a Zipfian distribution on the domain $[2,12]$ with group sizes of size 2
being most frequent.  Next, for each query, we generate a random value between $2$ and
$m$ as the low end of the group
size range in that query.  We then generate the query group size as $Min(2*(m-l)+1, 12-l+1)$, where 12 is the maximum
group size.  This serves to make $m$ the center of each group size
range in each prematched query.

In the rare case that we cannot find enough queries
to satisfy a group
signature, we generate a new group signature and continue. The process
terminates when a specified total number of queries have been generated. In order to
spread queries from pre-matched groups across the workload, we define a group interval
as a parameter that controls randomization. For example, if the
group interval is 5000, our workload generator will generate 5000 pre-matched groups
at a time and then randomly shuffle all queries within each range of 5000.

To serve as a yardstick, we designed an "optimal" algorithm (OPT) that is only
required to scan the workload once and find matches based on pre-matched group
id and group size information embedded in the query trace. This information, of course,
is ignored by the
regular matching algorithms.  The match percentage for the "optimal" solution is 100\%.
The average query latency for OPT varies when the group interval changes. However,
it should always be less than the group interval since a match is guaranteed to be
found within a group interval in OPT.

\subsection{Experimental Results}

In this section, we describe and interpret the experimental results to
identify the best performing heuristic algorithms.
Recall that BMA is the basic matching algorithm that processes the  set of
potential matches (CQS) in random order but stops after finding the first match;
LMF and HMF are also early stopping algorithms, but they process the CQS in the
order of least matchable (respectively highest matchable); finally NES is an
algorithm that processes the entire CQS even after finding an initial match to
search for larger matches.

We compare these algorithms in terms of system measures
(percentage of queries matched and query processing time), and user metrics (average
query latency and fairness measured by average group size).  Suppose, for example, that
a matching algorithm has a group size 2 query it can immediately
satisfy, but also a compatible group size 4 query that is compatible but not
immediately satisfiable.  The algorithm has
to exercise forbearance in order to satisfy the size 4 request in the future; greedy
choices, by contrast, can have good match percentage but poor fairness.

\begin{figure*}[t]
  \begin{minipage}[t]{0.29\textwidth}
\begin{center}
\scalebox{0.4}{\includegraphics{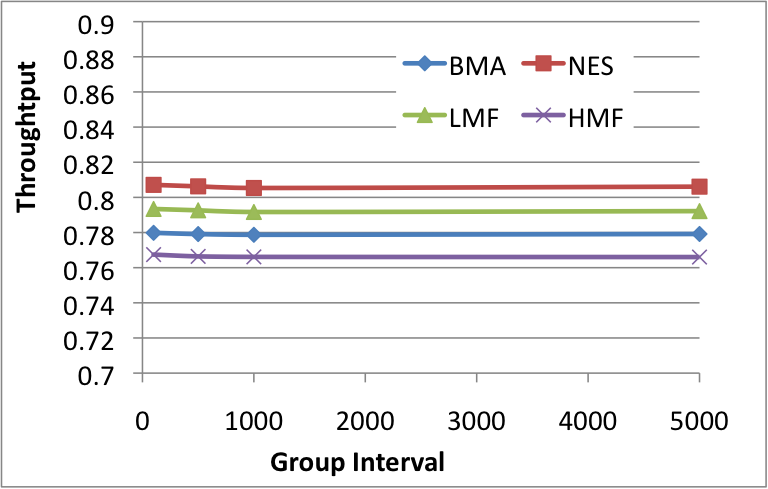}}
\end{center}
\vspace{-4mm}
\caption{Percentage of matched queries (throughput) as group interval increases.}
\label{fig-throughput}
  \end{minipage}
\hspace{5mm}
  \begin{minipage}[t]{0.27\textwidth}
\begin{center}
\scalebox{0.37}{\includegraphics{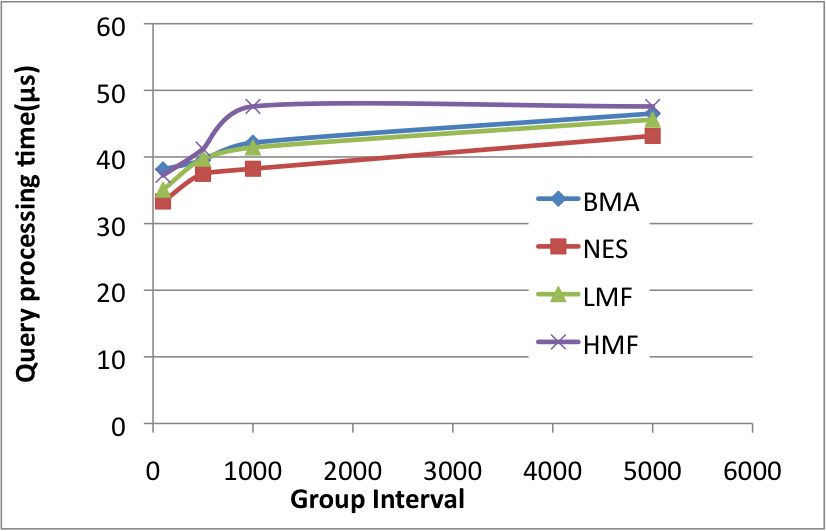}}
\end{center}
\vspace{-4mm}
\caption{Average query processing time as group interval increases.}
\label{fig-time}
  \end{minipage}
\hspace{5mm}
  \begin{minipage}[t]{0.27\textwidth}
\begin{center}
\scalebox{0.4}{\includegraphics{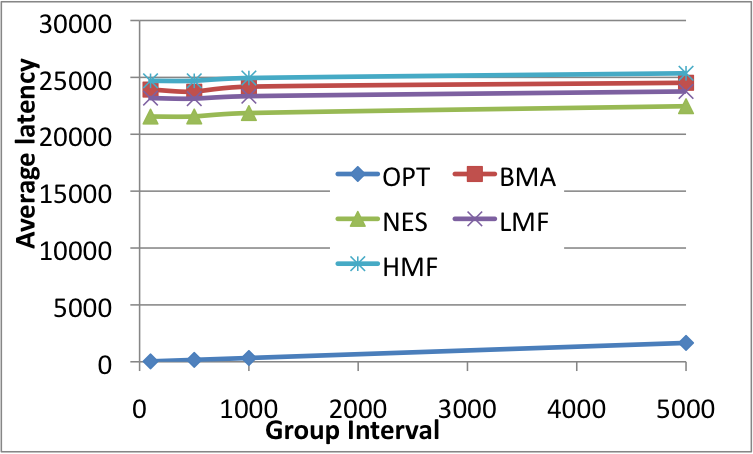}}
\end{center}
\vspace{-4mm}
\caption{Average query latency as group interval increases.}
\label{fig-latency}
  \end{minipage}
\end{figure*}

\begin{figure*}[t]
  \begin{minipage}[t]{0.29\textwidth}
\begin{center}
\scalebox{0.4}{\includegraphics{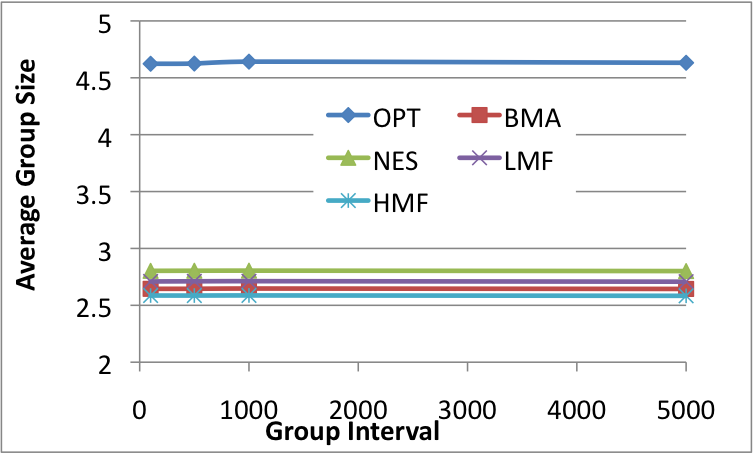}}
\end{center}
\vspace{-4mm}
\caption{Average matched group size as group interval increases.}
\label{fig-groupsize}
  \end{minipage}
\hspace{5mm}
  \begin{minipage}[t]{0.29\textwidth}
\begin{center}
\scalebox{0.4}{\includegraphics{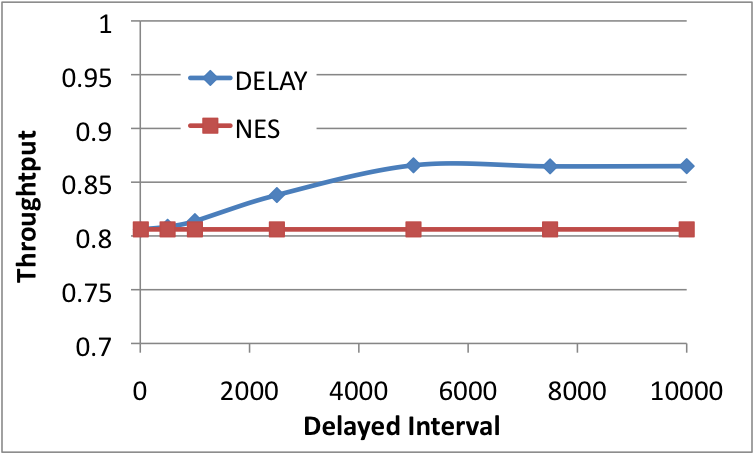}}
\end{center}
\vspace{-4mm}
\caption{Percentage of matched queries (throughput) as  delay interval increases.}
\label{fig-delaythroughput}
  \end{minipage}
\hspace{5mm}
  \begin{minipage}[t]{0.29\textwidth}
\begin{center}
\scalebox{0.4}{\includegraphics{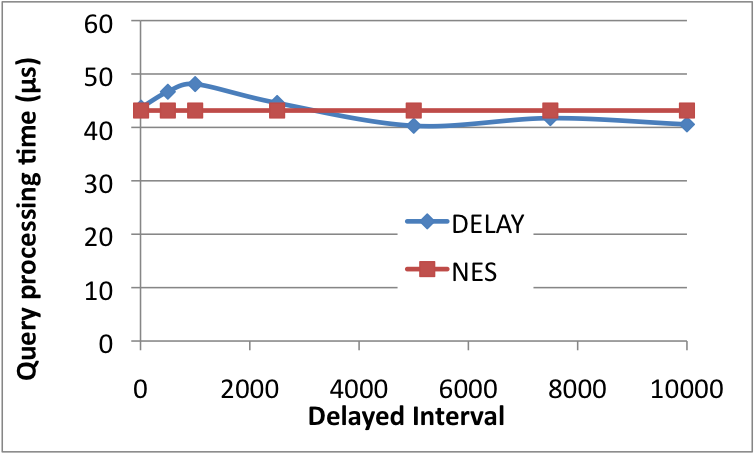}}
\end{center}
\vspace{-4mm}
\caption{Average query processing time as delay interval increases.}
\label{fig-delaytime}
  \end{minipage}
\end{figure*}

\begin{figure*}[t]
\centering
  \begin{minipage}[t]{0.4\textwidth}
\begin{center}
\scalebox{0.55}{\includegraphics{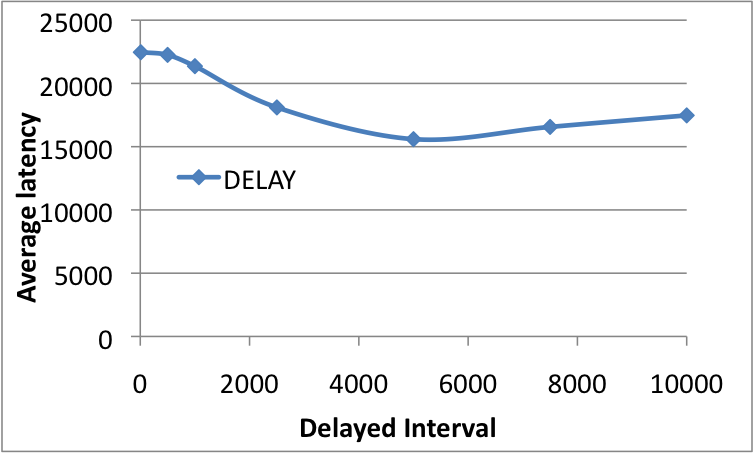}}
\end{center}
\vspace{-4mm}
\caption{Average query latency as delay interval increases.}
\label{fig-delaylatency}
  \end{minipage}
\hspace{5mm}
  \begin{minipage}[t]{0.4\textwidth}
\begin{center}
\scalebox{0.55}{\includegraphics{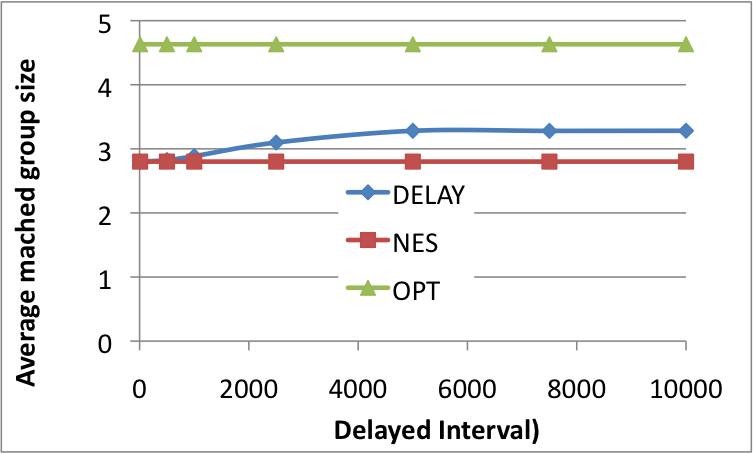}}
\end{center}
\vspace{-4mm}
\caption{Average matched group size as delay interval increases.}
\label{fig-delaygroupsize}
  \end{minipage}
\end{figure*}

{\em Performance with Varying Group Interval:}
Recall that the group interval is a parameter of the workload generator that
controls how scrambled the workload compared to a pre-matched set of queries that "seeds"
the workload. Clearly, an omniscient, optimal algorithm will be able to match all
queries in such a workload because this workload is a randomization of a workload
where there is  a match for all queries.    However, our algorithms such as BMA will do worse
than optimal because they are online and not offline, and make heuristic choices to reduce computation.
Figure~\ref{fig-throughput} shows the percentage of queries matched by our various
heuristics.  The hypothetical offline optimal algorithm is not shown because it always
achieves 100\% matching.

The figure shows that the percentage of matches does not change perceptibly as the
group interval (scrambling distance) increases for all algorithms.  It also shows
NES has the highest match percentage (around 81\%) while HMF performs worst (around 76\%)
and the baseline BMA is slightly better (77\%).  This is not surprising because
NES tries to find the best match without stopping.  Intuitively, HMF is bad
because it processes highly matchable queries early which makes it less
likely that later queries in the CQS will be matched.
Given early stopping, LMF does better because it does not squander more matchable
queries when other queries will do instead; a random order performs in between LMF and HMF.
Note that the difference between 76\% and 81\% may seem small but if there is revenue
attached to each match, a 5\% uptick in revenue is appreciable.

NES gets a higher match percentage by processing the CQS more thoroughly.
Thus we might hypothesize that NES will require more processing time in return for a
higher match percentage. Figure~\ref{fig-time} is perhaps surprising because
it shows that NES has the smallest average processing time (for example,
around 40 $\mu$sec at a group interval of 2000) while HMF has close to 50 $\mu$sec
and the others are in between).   Intuitively, this is because NES process
more queries on average per scan of  the CQS and hence removes queries early
from the CQS; this in turn requires less scanning overhead in the future.
Note also that the processing time increases slightly with scrambling distance;
this makes sense because the more far apart potentially matchable queries
are placed in the workload, the greater the average length of the CQS and
hence the processing time.

While NES does well from the system provider's point of view, what
about from the user point of view?  Figure~\ref{fig-latency} shows the
average latency (measured in terms of queries after which a user request is satisfied)
for the algorithms.   Recall that the workload is generated from a pre-matched set
of queries that are scrambled randomly within a group interval.  The latency of a query $Q$ for the optimal algorithm
is the difference between the index of the last query in the pre-matched set for
$Q$ and the index of $Q$ itself.  Thus the latency for the optimal algorithm shown as a
reference is very small and increases linearly with the group interval.  Once again,
NES has the lowest average latency (around 21,000 queries) and HMF has the worst
latency (around 25,000).

While NES is the best latency, it is much worse than optimal.   Further scrutiny
of the results revealed that 60\% of queries in NES are fast because NES does find
most pre-matched groups.  However, because NES is a heuristic, in
around 40\% of the cases, NES misses finding its pre-generated (and hence optimal)
groups.  In that case, either the query will never be matched or, by random chance
based on the workload, it can be matched later but after a much higher latency whose
average value is roughly the average time for a query to find another set of compatible queries.
Thus the average latency is pushed very high by outliers.   We calculated the median
latency for NES at a group interval of 5000 as 135 and the 90\% latency as around
7591 compared to the optimal latency of 1658.    Given that the outlier latency is
largely an artifact of the workload generation model, this artificial increased
latency is unlikely to be an issue in practice.

Finally, recall that in measuring fairness, we are trying to penalize
algorithms that boost match probability by artificially favoring some group
sizes (say small groups) over others.   Part of the novelty of our social
coordination system is that it allows coordination among users with various
size group requests (unlike say a dating service where the sizes are always 2).
It would be unfortunate if this generality in the service specification was
accompanied by service bias in terms of requested group size.  We use the average size
of the group when compared to the optimal algorithm as a fairness metric.
Figure~\ref{fig-groupsize} shows the average group size returned by the various algorithms.

Given the workload, the average group size of matches returned by the optimal algorithm
is around 4.6.  On the other hand, BMA is around 2.6 and NES is around 2.8.
This again follows because NES is prepared to wait to get larger groups and hence is
fairer to larger groups.

The results so far show that NES is better than using matchability as
a simple stopping criterion and certainly better than BMA with its random order
and early stopping.  However, despite this we believe that LMF is a promising heuristic and comes close to NES.   For example, if the CQS gets too large, NES may be
infeasible and the LMF heuristic may be needed at least for graceful degradation.

Further, LMF particularly shines when there are workloads with a
significant fraction of very discriminating users that are hard to satisfy (e.g., want to play
tennis in Campbell, CA at 6 am) and very easygoing users (willing to play tennis anywhere in the Bay Area at any time).  A small modification of Example 4 can be constructed by adding a fifth query
$q_5$ that can satisfy $q_1$ and $q_2$ such that all 5 queries are satisfied by LMF but only 3 are satisfied by NES.  Repeating this
sequence indefinitely leads to a workload where NES
achives a match percentage of only 60\% compared to 100\% for LMF.

{\em Performance of Delayed Evaluation:}
So far we have looked at the performance of the basic algorithm augmented with
matchability and no early stopping.
We now evaluate the effect of  our second and more complex refinement: adding a
fixed delay threshold.  Recall that in delayed evaluation, for each query,
we calculate the average age of other queries in its Compatible Set of Queries (CQS)
and delay the processing of that queries if the average age is too small based on
simple timer processing.  How does the basic performance of all measures
change as the delay increases?  Intuitively, if the group interval (scrambling interval)
is say 5000, delaying processing by around 5000 should make it more likely that a
query will behave close to its optimal and the gains should fall off after that.

This is indeed what we find.  In all cases, we fix the group interval at 5000
and use the NES algorithm augmented with a delay parameter that varies from
10 to 10,000. Figure~\ref{fig-delaythroughput} shows that as the delay
increases the percentage match increases from around 80\% for NES (reference line) to
around 87\% but the gains fall off after a delay of 5000.    Figure~\ref{fig-time}
shows that the query processing time can increase with delay from around 40 $\mu$sec to
around 50 $\mu$sec because of the overhead of the delay processing.  However, as we get
closer to the optimal delay of 5000, the processing time falls back to around 40 $\mu$sec
and is even slightly faster than NES.  We hypothesize that the better match efficiency
makes up for the increase in delay timer processing (which is a fixed overhead regardless
of the magnitude of the delay).

Figure~\ref{fig-delaylatency} shows a more interesting tradeoff.  As the delay
increases, we see that the latency drops sharply at first from 23,000 to around 15,000.
This is because increasing the match percentage by 5\% reduces the outliers by 5\%
which sharply decreases the average latency,  Beyond the optimal delay of 5000,
however, the match percentage does not increase and the delay is merely an artificial
waiting penalty: thus beyond 5000 the average latency starts increasing again.

The user will rightly complain that this is cheating.  In a synthetic workload, we
knew the group interval and hence could estimate the optimal delay.  However, in a
real deployment, the system can keep statistics such as
Figure~\ref{fig-delaylatency} by periodically varying the delay (akin to explore-exploit
systems) to find the knee of the curve. The knee of the
curve  can be used to estimate  the optimal amount of delay to be added.

Finally, Figure~\ref{fig-delaygroupsize}  shows that delayed processing also
decreases the bias against large groups by coming closer to the average group
size of the optimal algorithm.  As the delay increases to the knee of the curve (5000),
the average group size gets close to 3.5 which is much closer to 4.5 which
represents the ideal (optimal) compared to say 2.6 for BMA.

{\em Recommendations for Deployment:}
Our results suggest that a combination of DELAY and NES performs well across all measures. The system can find the smallest value of added delay after which matching performance does not improve significantly. LMF may be useful in workloads that have many discriminating users, and when the average CQS length become so long that NES becomes infeasible.

In terms of scaling, an average of 40 $\mu$sec implies 25,000 queries a second on a single core using a high end 2.5 Ghz machine.  However, throughput can be increased by forking multiple threads, leveraging multiple cores, and parallelizing query processing by region using multiple servers: users in Cairo are unlikely to coordinate with users in Cupertino.

\section{Conclusions}
\label{sec:conc}
Social {\em coordination} may represent the next step beyond the social {\em awareness} provided by today's Online Social Networks (OSNs).  We focused on the problem of fluid social coordination where the set of people involved is unknown at the start, may change quickly over time, and may {\em not} be part of one's friends in any OSN. Such fluid coordination allows  forming weak ties~\cite{granovetter1973} that can enrich our lives beyond the strong ties formalized by OSNs.

In our formalism, users declaratively specify coordination objectives using selection constraints on coordination variables, join constraints on pairs of user variables, and group size constraints.  While declarative specification can reduce user effort compared to browsing, we recognize that, in some cases, browsing can allow more flexibility.  We suggest the following research directions in social coordination:

{\em 1. Economics:} Enmeshed queries extend dating services to arbitrary group sizes.   Some dating services provide better matches for users who pay more.   What is the natural way to extend enmeshed queries to specify a willingness to pay for matches?  How does this relate to auction theory?

{\em 2. Soft Constraints:} While we consider hard coordination constraints, users might prefer to declare relative preferences over attributes like place and time. How can enmeshed queries be extended to handle such ``soft-constraints''?

{\em 3. Recommendations:} Amazon and NetFlix recommend new choices based on past selections.  On what basis should a coordination system recommend groups to users?  A user that often plays tennis may like to hear about 3 nearby, compatible tennis players who wish to find a fourth player.

{\em 4. Query Flexibility:}  Our simple model of a conjunction of disjunctions can be generalized.  For example, in ad matching~\cite{fontoura10} a richer set of boolean queries can be expressed at the possible cost of increased computation time.   Further, our paper uses a completely declarative (system chooses) model and allows no browsing (user chooses) as in Meetup.  What is the best way to combine browsing and declaration?

{\em 5. Data Mining:}  Social coordination tools may allow social scientists to answer questions about the sociology of coordination.  What are the metrics (e.g., fraction of successful coordinations, keystrokes to coordination completion) by which one judges a successful social coordination?  A social scientist could test hypotheses about successful coordinations such as ``small groups are more likely to find matches".

Social coordination fulfills a basic human need to connect to other human beings.  Further, fluid social coordination allows building weak ties to a larger set of people than our friends. While such weak ties classically occur by serendipity as in meetings at the proverbial water cooler, enmeshed queries may help institutionalize such serendipity.  While we have taken a small step in formalizing aspects of the problem, the vista for further work appears inviting: the solution space can be enriched if social scientists, database theorists, economists, and system designers all join the conversation.

\section{Acknowledgments}
We would like to thank Adam Silberstein for inspiring our work on enmeshed queries.

\bibliographystyle{abbrv}
\bibliography{enmeshed}

\end{document}